\documentclass[11pt,a4paper]{article}
\usepackage[english]{babel}
\usepackage[normalem]{ulem}
\usepackage[all]{xy}
\usepackage{amsfonts,amsmath,amssymb,amsthm,booktabs,color,dsfont,epsfig,graphicx,mathrsfs,multirow,scalefnt}

\newtheorem{theo}{Theorem}[section]
\newtheorem{cor}[theo]{Corollary}
\newtheorem{lem}[theo]{Lemma}

\newtheorem{prop}[theo]{Proposition}
\newtheorem{defn}[theo]{Definition}
\newtheorem{rem}[theo]{Remark}

\newcommand{\acknowledgement}[1]{\par\addvspace\baselineskip\noindent\textbf{\small{Acknowledgement.}}\enspace\ignorespaces\small#1}

\newcommand{\X}{\mathcal{X}}

\newcommand{\Y}{\mathcal{Y}}

\newcommand{\R}{\mathbb{R}}

\newcommand{\keywords}[1]{\par\addvspace\baselineskip\noindent\textbf{Keywords:}\enspace\ignorespaces#1}
\newcommand{\AMSclassification}[1]{\par\addvspace\baselineskip\noindent\textbf{Mathematical subject classification:}\enspace\ignorespaces#1}

\title{A nonsmooth two-sex population model}

\author{
\small{Eduardo Garibaldi}\\
\footnotesize{UNICAMP -- Department of Mathematics}\\
\footnotesize{13083-859 Campinas - SP, Brazil}\\
\footnotesize{\texttt{garibaldi@ime.unicamp.br}}
\and
\small{Marcelo Sobottka}\\
\footnotesize{UFSC -- Department of Mathematics}\\
\footnotesize{88040-900 Florian\'opolis - SC, Brazil}\\
\footnotesize{\texttt{sobottka@mtm.ufsc.br}}
}

\date{}

\begin{document}

\maketitle

\begin{abstract}
This paper considers a two-dimensional logistic model to study populations with two genders. The growth behavior of a population is guided by two coupled ordinary differential equations given by a non-differentiable vector field whose parameters are the secondary sex ratio (the ratio of males to females at time of birth), inter-, intra- and outer-gender competitions, fertility and mortality rates and a mating function. For the case where there is no inter-gender competition and the mortality rates are negligible with respect to the density-dependent mortality, using geometrical techniques, we analyze the singularities and the basin of attraction of the system, determining the relationships between the parameters for which the system presents an equilibrium point. In particular, we describe conditions on the secondary sex ratio and discuss the role of the average number of female sexual partners of each male for the conservation of a two-sex species.

\keywords{population dynamics, two-sex models, nonsmooth ordinary differential equations}

\AMSclassification{34C60, 37C10, 37N25, 92D25}
\end{abstract}

\bigskip
\hrule
\noindent
{\footnotesize\em This is a pre-copy-editing, author-produced preprint of an article accepted for publication in Mathematical Bioscience. The definitive publisher-authenticated version Eduardo Garibaldi and Marcelo Sobottka, A nonsmooth two-sex population model, Mathematical Bioscience, Volume 253, 15 July 2014, Pages 1-10, ISSN 0025-5564, is available online at: http://dx.doi.org/10.1016/j.mbs.2014.03.015 or\break
http://www.sciencedirect.com/science/article/pii/S0025556414000728 .}
\hrule
\bigskip

\section{Introduction}

When studying biological populations in nature, it is usual to recognize an unvarying proportion of the genders in a stable environment.
Such a prevalent observation has been a remarkable motivation for fundamental contributions in the theory of sex-structured populations.
Fisher's comprehension \cite{Fisher} of the commonness of nearly 1:1 sex ratios, Hamilton's explanation \cite{Hamilton} for the existence of
biased sex ratios, Trivers-Willard hypothesis \cite{TriversWillard} on the parental capability to adjust the sex ratio of offsprings as a response
to environmental changes and Charnov mathematical proposal \cite{Charnov} for sex allocations are some relevant examples of this kind of legacy.

In a previous work \cite{Garibaldi-Sobottka}, we have developed a dynamic-programming model in order to discuss whether the identification
of a stable sex ratio in nature might mirror a population maintenance cost under finite resources. Here we propose another
dynamical approach to study sex-structured populations which consists in modeling the time evolution of two-sex populations with
differential equations. Under this point of view, the interactions of the individuals are represented as a mean tendency of the
whole population. Furthermore, instead of looking for a sex ratio that would maximize the efficiency of individuals in the use of available resources, in the population-dynamics formulation, secondary sex ratio is actually one of the parameters of the system. In such a case, the aim is thus,
for suitable mating functions, to describe and classify the behavior of the population for distinct progeny sex ratios and distinct
mortality sex ratios \cite{Fredrickson,Goodman,Hoppensteadt,Ranta_et_Al,Rosen83,YangMilner,YellinSamuelson}.
For instance, it has been argued in \cite{Ranta_et_Al,YellinSamuelson} that the marriage rate plays an important role in the stability of the population, since polygamy would amplify the sensibility of the system to the variation of the other parameters. In another direction, a model
with stable solutions for monogamous and polygamous populations was presented in \cite{Rosen83}.

In this paper, we propose a nonsmooth two-sex logistic model (which may be seen as an extension of previous formulations) and we use the qualitative-geometric theory of ordinary differential equations to study it. Considering sex-ratio dependent competition terms, we obtain sufficient and necessary conditions for the persistence of the population. In particular, we show that the dynamical behavior of the population is governed by a highly nonlinear relationship between the secondary sex ratio and the competition parameters, and that the average number of male's reproductive partners is an important parameter that may allow a two-sex species to find a stable equilibrium.

The paper is organized as follows. In section \ref{sec.model}, we recall some classical models for two-sex populations and we
define the model that will be studied. In section \ref{sec.singularities}, we detail its singularities by analyzing two vector fields defined on the plane and naturally associated to the original one. In section \ref{sec.sex_ratios},
we study the relationships between secondary and tertiary sex ratios and the competition parameters of the model. In section \ref{sec.behavior}, we
describe the local and global behavior of the two associated vector fields. In particular, we point out conditions on the secondary sex ratio that assure
the existence of asymptotically stable singularities and the nonexistence of cycles. Hence, we discuss the local and global dynamics for the original vector
field. In section \ref{sec.final}, we outline open questions about the dynamics of the model and some possible extensions.

\section{The model}\label{sec.model}

We consider here a two-sex logistic model which follows the basic lines of the classical logistic model: the population growth is given by the balance between the birth rate (which depends on the quantity of individuals in the population) and the death rate (which depends on square of the quantity of individuals, representing the interactions between them).

Non-logistic models for two-sex populations have been proposed at least since the 1940's (for a review see \cite{HWW}). For instance, Kendall (\cite{Kendall49}, page 247) proposed two non-logistic models. The first one consists in a model for the behavior of male and female populations described by the following coupled ODE's:
\begin{equation}\label{kendall}\begin{array}{lcl}
\dot{x}&=&b_x F(x,y)-m_x x\\\\
\dot{y}&=&b_y F(x,y)-m_y y,
\end{array}
\end{equation}
where $x$ and $y$ denote the quantity of females and males at time $t$, respectively, $m_x$ and $m_y$ denote the mortality rates of females and males, $b_x$ and $b_y$ are independent parameters for the birth rate of each gender, and $F$ is the mating function (which was supposed to be nonnegative and symmetric in $x$ and $y$) and represents the contribution of males and females to the birth rate. In his work, Kendall studied the case where $m_x=m_y$, $b_x=b_y=1/2$ and $F$ has one of the following forms:
$$xy,\quad (xy)^{1/2},\quad x+y\quad \text{or}\quad  \min\{x,y\}.$$

The second model proposed by Kendall addresses the problem of pair formation in two-sex populations. In such a model, three coupled
ODE's take into account the numbers of unmarried males, unmarried females and married couples. Once again, a central role is played by the mating function.

Following the Kendall's work, Goodman \cite{Goodman} studied the cases where $m_x \neq m_y$ and $b_x \neq b_y$ for several mating functions, including the above ones as well as $F(x,y)=x$ and $F(x,y)=y$. In \cite{Fredrickson}, Fredrickson
assumed two hypotheses on the mating function $F$: {\em heterosexuality} (that is, $F(0,y)=F(x,0)=0$) and {\em homogeneity} (in the sense that $F(kx,ky)=kF(x,y)$). Using these hypotheses, he found a general form for differentiable mating functions and deduced that they are {\em consistent}: if there is a preponderance of some gender in the population, then the birth rate will be limited by the number of individuals of the other gender.
Other natural hypothesis on $F$ is {\em monotocity} \cite{YangMilner}, namely, if $\bar{x}\geq x$ and $\bar{y}\geq y$, then $F(\bar{x},\bar{y})\geq F(x,y)$.

Logistic models for two-sex populations have been considered by the academic community \cite{ChavezHuang,Rosen83,YangMilner}. The model in \cite{ChavezHuang} incorporates nonlinear birth and separation processes to Kendall's pair-formation model, while the model studied in \cite{YangMilner} is an age-dependent two-sex model with density dependence in the birth and death. On the other hand, Rosen (\cite{Rosen83}, section 4) studied a model which admits in \eqref{kendall} terms for competition:
\begin{equation}\label{rosen}\begin{array}{lcl}
\dot{x}&=&b_x F(x,y)-(m_x x+X_x x^2+X_{xy}xy)\\\\
\dot{y}&=&b_y F(x,y)-(m_y y+Y_y y^2+Y_{xy}xy),
\end{array}
\end{equation}
where $X_x$ and $Y_y$ describe the effects of intrasexual competition of females and males, respectively, and $X_{xy}$ and $Y_{xy}$
characterize the intersexual competition of males on females and females on males, respectively. Furthermore, in \cite{Rosen83} it was considered the mating function given by
\begin{equation}\label{matingfunction}
F(x,y)= \min\{x,r y\},
\end{equation}
where $r$ is the average number of female sexual partners that each male has along each reproductive cycle ($r<1$ may be interpreted as polyandrous population, $r=1$ is understood as a monogamous population, and $r>1$ may be seen as a polygynous population).

We recall that models like~\eqref{kendall} and~\eqref{rosen} do not inspect in an explicit way certain internal mechanisms of the populations, like
pair formation or age structure. In fact, such mechanisms are captured by the parameters of the models. Consider, for example, a population of a total of $m$ males and $f$
females of which $\tilde{m}$ males are sexually active and $\tilde{f}$  females are receptive and each one of them has fertility rate $\tilde{s}$. Suppose yet that each sexually
active male successfully breeds with $\tilde{r}$ females. In such a case, these models will interpret that all the males successfully breeds with $r=(\tilde{m}\tilde{r})/m$ females and all the females are receptive (each one of them with fertility rate $s$), so the net number of individuals being born and the magnitude of competitions will be virtually the same and the models will reveal the behavior of the population growth. Notice that competitions for mating are not focused by these models, since they are part of the pair formation mechanism and in general they do not affect the mortality rate. Besides, since the parameters $ r $ and $ s $ absorb the age structure and pair formation, the sex-ratio type considered in the models is the tertiary sex ratio
(the number of adult males divided by the number of adult females -- also named adult sex ratio), which, when adopting such a point of view, is
indistinguishable from the operational sex ratio (the number of sexually active males divided by the number of receptive females).

Note that competition terms of the form $xy$ in the above equations may not capture some aspects of the relationship between the genders. In fact, although for predator-prey models it is reasonable to suppose that a great number of predators or preys will increase the probability of interactions between the species and then the population growth of both species will be affected by the quantity $xy$, this interpretation does not necessarily hold for two-sex populations, in which one of the genders is not a vital resource but in general both genders coexist and have common resources. The causes and consequences of adult sex ratio and operational sex ratio have been extensively investigated by biologists. There are pieces of evidence that the sex ratio has impact on fitness prospects of males and females and on optimal sex allocation decisions \cite{Michler_et_al11}. It was noticed that a male-biased sex ratio could amplify male-male
competition with negative impact on female survival and fecundity (see, for instance, \cite{Galliard_et_al, Grayson_et_al}).
The excess of males against females has also been pointed out as a likely negative factor for females in the human case \cite{HeskethXing}. In other words, when $y$ is much greater than $x$, even if $xy$ is small, one may detect a negative impact on the $ x $ population. These observations lead us to incorporate to equations \eqref{rosen} a mortality term for each gender which takes into account the ratio between the genders on the intersexual competitions. In other words, the female population will have a mortality term proportional to $\frac{y}{x}xy=y^2$, while the male population will have a mortality term proportional to $\frac{x}{y}xy=x^2$  (these terms can be seen as `outer-gender competition terms').

Let us denote
\begin{description}

\item[$\underline{\mathbf{1}_{\R_+^*}}$] the characteristic function of the set $\R_+^*:=(0,+\infty)$;

\item[\uline{$s$}] the average birth rate of population;

\item[\uline{$\rho \in (0,1)$}] the average percentage of female births per pregnancy (thus $1-\rho$ indicates the average percentage of male births, while $(1-\rho)/\rho$ is the secondary sex ratio of the population);

\item[\uline{$s \mu_x, s \mu_y$}] the mortality rate for females and males, respectively;

\item[$\underline{s\X_x, s\X_y \ge 0}$] indicating how the growth of the female population is negatively affected by its own size and by the size of the male population, respectively. We suppose that $ \X_x + \X_y > 0 $, since otherwise there would not be a coercive force to limit the population growth of the females and either both genders would have an unlimited growth or the female population would increase until the male population would become extinct;

\item[$\underline{s\Y_x, s\Y_y \ge 0}$] describing how the growth of the male population is negatively affected
by the size of the female population and by its own size, respectively. As before, at least one of these parameters will be strictly positive;

\item[$\underline{s\X_{xy}, s\Y_{xy} \ge 0}$] indicating how the growth of the female population is negatively affected
by the interaction with the male population, and how the growth of the male population is negatively affected
by the interaction with the female population, respectively.
\end{description}
We present then the model
\begin{equation}\label{Model}\begin{array}{lcl}
\dot{x} &=&\displaystyle \rho s F(x,y) - s(\mu_x x+\X_x x^2 + \X_{xy}xy+\X_y\mathbf{1}_{\R_+^*}(x) y^2), \\\\
\dot{y} &=&\displaystyle (1-\rho) s F(x,y) - s(\mu_y y+ \Y_x\mathbf{1}_{\R_+^*}(y) x^2 + \Y_{xy}xy+ \Y_y y^2),
\end{array}\end{equation}
where $F$ is the mating function given by \eqref{matingfunction}. Note that the terms $\X_y y^2$ and $\Y_x x^2$ above are multiplied by
 $\mathbf{1}_{\R_+^*}(x)$ and $\mathbf{1}_{\R_+^*}(y)$, respectively, since the effect of a gender on the other one will only be considered in its presence.

In the next sections, we will study the behavior of the system for the situation where there is no intersexual competition and the mortality rate of each
gender is negligible with respect to the density-dependent mortality.

\section{Singularities of the vector fields}\label{sec.singularities}

From now on, we will consider $\X_{xy}=\Y_{xy}=\mu_x=\mu_y=0$. Furthermore, due to our qualitative and geometric approach, without loss of generality, we may assume $ s = 1 $ in our analysis.
Using \eqref{Model}, let then $\Phi:(\R_+)^2\to\R^2$ be the vector field given by
\begin{equation*}\Phi(x,y):=\bigl(\dot{x}(x,y),\dot{y}(x,y)).\end{equation*}

Note that, defining the maps $\Phi_I:\R^2\to\R^2$ and $\Phi_{II}:\R^2\to\R^2$ by
\begin{equation*}\begin{array}{lcl}
\Phi_I(x,y):=\Big(\rho x - (\X_x x^2 + \X_y y^2)\ ,\ (1-\rho) x - (\Y_x x^2 + \Y_y y^2)\Big),\\\\
\Phi_{II}(x,y):=\Big(\rho ry - (\X_x x^2 + \X_y y^2)\ ,\ (1-\rho) ry - (\Y_x x^2 + \Y_y y^2)\Big),\end{array}
\end{equation*}
and the regions
\begin{equation*}
R_I:=\{(x,y)\in (\R_+^*)^2:\ y-r^{-1}x\geq 0\}, \qquad R_{II}:=\{(x,y)\in (\R_+^*)^2:\ y-r^{-1} x \leq 0\},
\end{equation*}
we have that, except on the axes, $\Phi$ can be written as
\begin{equation}\label{Phi_I_II}\Phi(x,y)=\left\{\begin{array}{ll}
\Phi_I(x,y) &\text{ if } (x,y)\in R_I\\\\
\Phi_{II}(x,y) &\text{ if } (x,y)\in R_{II}
\end{array}\right..\end{equation}

Therefore, a strategy to understand the flow generated by the vector field $\Phi$
consists in studying the flows generated by the vector fields $\Phi_I$ and $\Phi_{II}$,
and the way as they are coupled along the ray $x=ry$, $y \ge 0$.

\subsection{Singularities of the vector fields $\Phi_I$ and $\Phi_{II}$}

In the sequence, we will study the singularities of the vector fields $\Phi_I$ and $\Phi_{II}$.
The existence of singularities for the vector fields and the type of these singularities obviously
depend on the choice of parameters.

\begin{defn}
We denote
\begin{equation*}
\Delta:=\X_x\Y_y-\X_y\Y_x, \quad \Delta_y:=\rho\Y_y-(1-\rho)\X_y, \quad \text{and} \quad \Delta_x:=\rho\Y_x-(1-\rho)\X_x.
\end{equation*}
\end{defn}

Using the above notation, we have the following result.

\begin{theo}\label{VFI_Theo}
The vector fields $\Phi_I:\R\to\R$ and $\Phi_{II}:\R\to\R$ admit finitely many singularities with non-null coordinates if, and only if,
\begin{equation}\label{NonNullDenominator}
\Delta\neq 0,
\end{equation}
and
\begin{equation}\label{negative}
\Delta_x\Delta_y< 0.
\end{equation}
Moreover, under the above conditions, $\Phi_I:\R\to\R$ has the singularities $(0,0)$, $(x_I,y_I)$ and $(x_I,-y_I)$, where
\begin{equation}\label{simp_sing_I}
x_I=\frac{\Delta_y}{\Delta}, \qquad \qquad y_I=\frac{\sqrt{- \Delta_x\Delta_y}}{|\Delta|},
\end{equation}
while $\Phi_{II}:\R\to\R$ has the singularities $(0,0)$, $(x_{II},y_{II})$ and $(-x_{II},y_{II})$, where
\begin{equation}\label{simp_sing_II}
x_{II}=\frac{\sqrt{-\Delta_x\Delta_y}}{|\Delta|}r,  \qquad \qquad y_{II}=-\frac{\Delta_x}{\Delta}r.
\end{equation}
\end{theo}

\begin{proof}
We will only prove the result for the vector field $\Phi_I$, since the proof for $\Phi_{II}$ is completely analogous.
Let us show that \eqref{NonNullDenominator} and \eqref{negative} are sufficient conditions to have finitely many singularities with non-zero coordinates and that the singularities are given accordingly to \eqref{simp_sing_I}. So let $(x_I,y_I)$ be a non-null solution of
\begin{equation}\label{EqSingI}
\Big(\rho x - (\X_x x^2 + \X_y y^2), (1-\rho) x - (\Y_x x^2 + \Y_y y^2)\Big) = (0, 0).
\end{equation}

For a moment, suppose that $\X_y$ and $\Y_y$ are both non-null. By solving the first equation for $y^2$ and then using it in the second equation, we get that $x_I$ satisfies
\begin{equation*}
\left(\frac{\Y_x}{\Y_y}-\frac{\X_x}{\X_y}\right)x_I=\frac{1-\rho}{\Y_y}-\frac{\rho}{\X_y}.
\end{equation*}
Since $\frac{\Y_x}{\Y_y}-\frac{\X_x}{\X_y}=-\frac{\Delta}{\X_y\Y_y}\neq 0$ and $\frac{1-\rho}{\Y_y}-\frac{\rho}{\X_y}=-\frac{\Delta_y}{\X_y\Y_y}\neq 0$, we obtain that $x_I=\frac{\Delta_y}{\Delta}$.

Thus, replacing the value of $x_I$ in the expression obtained for $y^2$, we find out that the nonnegative second coordinate will be
\begin{equation*}
y_I = \sqrt{x_I\left(\frac{\rho}{\X_y}-\frac{\X_x}{\X_y} x_I\right)}
=\sqrt{\frac{\Delta_y}{\Delta}\left(\frac{\rho}{\X_y}-\frac{\X_x}{\X_y} \frac{\Delta_y}{\Delta}\right)}=\frac{\sqrt{- \Delta_x\Delta_y }}{|\Delta|}.
\end{equation*}

Now, observe that if $\X_y = 0$ or $\Y_y = 0$, then the above solution obtained for $\eqref{EqSingI}$ can be also achieved by equaling the
corresponding parameter(s) to zero in \eqref{simp_sing_I}. Thus, we have proved that \eqref{NonNullDenominator} and \eqref{negative} are
sufficient conditions to have finitely many non-null singularities.

To see that \eqref{NonNullDenominator} and \eqref{negative} are also necessary conditions, the reader may check without difficulty that,
if $\Delta= 0$ but $\Delta_y\neq0$ or $\Delta_x\neq 0$, then the unique singular point for $\Phi_I$ and $\Phi_{II}$ is the origin.
Finally, notice that, if $\Delta_x=\Delta_y= 0$, then $ \Delta = 0 $ and the singularities of $\Phi_I$ are all the points
belonging to the conic defined by $\X_x x^2 + \X_y y^2 - \rho x=0$, while the singularities of $\Phi_{II}$ are all
the points belonging to the conic defined by $\X_x x^2 + \X_y y^2 - \rho ry=0$.
\end{proof}

\begin{rem}
For the degenerate case $ \Delta = \Delta_x = \Delta_y = 0 $, it is easy to see that
each singularity is nonhyperbolic; however, since this is a non-generic situation, we will not treat it in this work.
\end{rem}

\section{Sex ratios}\label{sec.sex_ratios}

We regroup in this section several results on dynamical properties related to sex ratios which will be useful in the local and global analysis
of the proposed nonsmooth system.

\begin{defn} The secondary sex ratio of the population is the quantity $\sigma:=\frac{1-\rho}{\rho}$.
\end{defn}

\begin{defn} Let $\bar x$ and $\bar y$ be, respectively, the female and the male populations at equilibrium (whenever it exists). Then, the tertiary sex ratio of the population is the quantity $\tau(\bar x,\bar y):=\frac{\bar y}{\bar x}$. When $(\bar x,\bar y)$ is the unique equilibrium point we will denote the tertiary sex ratio simply by $\tau$.
\end{defn}

Under the convention that a division of a positive number by zero is $+\infty$, we can assure a singularity in the first quadrant for each vector field
by using the conditions below, which compare the ratios of competition factors to the secondary sex ratio.

\begin{prop}\label{VFI_Cor} If either
\begin{equation}\label{OUOU_1}
\frac{\Y_x}{\X_x}<\sigma<\frac{\Y_y}{\X_y}
\end{equation}
or
\begin{equation}\label{OUOU_2}
\frac{\Y_y}{\X_y}<\sigma<\frac{\Y_x}{\X_x},
\end{equation}
then $\Phi_I$ and $\Phi_{II}$ have exactly three distinct singularities given by theorem \ref{VFI_Theo}, besides both $(x_I,y_I)$ and $(x_{II},y_{II})$ belong to the first quadrant. If $\frac{\Y_x}{\X_x}=\sigma=\frac{\Y_y}{\X_y}$, then $\Phi_I$ and $\Phi_{II}$ have infinitely many singularities which are the points belonging to the conic defined, respectively, by $\X_x x^2 + \X_y y^2 - \rho x=0$ and $\X_x x^2 + \X_y y^2 - \rho ry=0$. In any other case, neither $\Phi_I$ nor $\Phi_{II}$ have singularities on the first quadrant.
\end{prop}

\begin{proof}
Note that \eqref{OUOU_1} is equivalent to
$\Delta>0 $, $ \Delta_y> 0 $ and $ \Delta_x< 0 $,
while \eqref{OUOU_2} is equivalent to
$ \Delta<0 $, $ \Delta_y< 0 $ and $ \Delta_x> 0 $.
Therefore, from the previous theorem, we immediately get the result.
\end{proof}

\begin{rem} Note that if \eqref{OUOU_1} holds, then neither $\X_x$ or $\Y_y$ are null, while if \eqref{OUOU_2} holds, then neither $\X_y$ or $\Y_x$ are null. However, in this study, we will only treat generic cases,
so from now on we will assume that all these parameters are strictly positive.
\end{rem}

The next result shows that both equilibrium points $(x_I,y_I)$ and $(x_{II},y_{II})$ (each one with respect to its respective vector field)
correspond to populations with the same tertiary sex ratio. In particular, it means that $(x_I,y_I)$ and $(x_{II},y_{II})$ are collinear
with the origin, so they belong to the same region $R_I$ or $R_{II}$.

\begin{theo}\label{Rel_bet_sing}
Suppose that either \eqref{OUOU_1} or \eqref{OUOU_2} holds. Then $\tau(x_I,y_I)=\tau(x_{II},y_{II})=:\tau$.
In particular, $(x_{II},y_{II})=(x_I,y_I)r\tau$.
\end{theo}

\begin{proof}
Notice that directly from \eqref{simp_sing_I} and \eqref{simp_sing_II}, we get that
$ \frac{y_I}{x_I}=\frac{y_{II}}{x_{II}}=\sqrt{\left|\frac{\Delta_x}{\Delta_y}\right|}=:\tau $
and $ x_{II}=ry_I $.
Therefore, $x_{II}=ry_I= r\tau x_I$ and $ y_{II}=\tau x_{II}=r\tau y_I $.
\end{proof}

We have then immediate corollaries.

\begin{cor}\label{tau(sigma)}
If either \eqref{OUOU_1} or \eqref{OUOU_2} holds, then (for the vector field $\Phi_I$ as well as for the vector field $\Phi_{II}$) the tertiary sex ratio is given by
\begin{equation}\label{tau}\tau=\sqrt{\left|\frac{\Delta_x}{\Delta_y}\right|}=\sqrt{\frac{\sigma\X_x-\Y_x}{\Y_y-\sigma\X_y}}.\end{equation}
\end{cor}

\begin{cor}\label{Cor.Rel_bet_sing}
Under the same assumptions of theorem \ref{Rel_bet_sing}, we have that
\begin{align*}
r^{-1}<\tau & \quad \Leftrightarrow \quad x_I<x_{II}\quad\text{and}\quad y_I<y_{II},\\
r^{-1}=\tau & \quad \Leftrightarrow \quad x_I=x_{II}\quad\text{and}\quad y_I=y_{II},\\
r^{-1}>\tau & \quad \Leftrightarrow \quad x_I>x_{II}\quad\text{and}\quad y_I>y_{II}.
\end{align*}
\end{cor}

\begin{defn}\label{comp.pol} For $ \X_x, \X_y, \Y_x, \Y_y > 0 $, the competition polynomial is defined by
$$Q(a):=\X_ya^3-\Y_y a^2+\X_x a-\Y_x.$$
\end{defn}

\begin{lem}
If $\Delta\neq 0$, then the polynomial $Q$ has exactly one real root which lies in the open interval with endpoints $\frac{\Y_x}{\X_x}$ and $\frac{\Y_y}{\X_y}$.
\end{lem}

\begin{proof}
Consider the real functions $f(a):=a^2$ and $g(a):=\frac{a\X_x-\Y_x}{\Y_y-a\X_y}$.
Thus, $\alpha$ is a real root of $Q$ if, and only if, $f(\alpha)=g(\alpha)$.
Hence, the result follows directly from a graphical analysis of $f$ and $g$, considering the cases $\frac{\Y_x}{\X_x}<\frac{\Y_y}{\X_y}$ and $\frac{\Y_x}{\X_x}>\frac{\Y_y}{\X_y}$.
\end{proof}

\begin{cor}\label{AlphaSigmaTau}
Suppose either \eqref{OUOU_1} or \eqref{OUOU_2} holds. We have $\tau=\sigma$ if, and only if, $\sigma$ is the real root of $Q$. Furthermore, if $\alpha$ is the real root of $Q$, then
\begin{description}
\item[\qquad under \eqref{OUOU_1}:] $\tau$ is a increasing function of $\sigma$ and either $\tau\leq\sigma\leq\alpha$ or $\tau\geq\sigma\geq\alpha$; \item[\qquad under \eqref{OUOU_2}:] $\tau$ is a decreasing function of $\sigma$ and either $\sigma\leq\alpha\leq\tau$ or $\sigma\geq\alpha\geq\tau$.
\end{description}
\end{cor}

\begin{proof}
Under either \eqref{OUOU_1} or \eqref{OUOU_2}, we have $\Delta\neq 0$ and $\Delta_x/\Delta_y<0$. Therefore, from corollary \ref{tau(sigma)}, $$\tau^2=\frac{\sigma\X_x-\Y_x}{\Y_y-\sigma\X_y}=g(\sigma),$$
and thus $\tau^2=\sigma^2$ means $g(\sigma)=f(\sigma)$, which is equivalent to $Q(\sigma)=0$.
The second part of the corollary follows easily from a graphical analysis of $f$ and $g$, considering the cases \eqref{OUOU_1} and \eqref{OUOU_2}.
\end{proof}

Note that the previous proposition says that populations with equal secondary and tertiary sex ratios are in fact non-generic cases.
Moreover, while the secondary sex ratio has bounds given either by condition \eqref{OUOU_1} or by condition \eqref{OUOU_2},
the tertiary sex ratio can assume, \emph{a priori}, any positive value.

The next proposition states that if $x$ and $y$ are positive and sufficiently near the origin, then the vectors $\Phi_I(x,y)$ and $\Phi_{II}(x,y)$
have slopes near to $\sigma$. The proof is straightforward and will be omitted.

\begin{prop}\label{VetcorFieldAtOrigin}
If $m>0$, then ${\displaystyle \lim_{x\to 0^+} \frac{\Phi_I(x,mx)}{||\Phi_I(x,mx)||}=\lim_{x\to 0^+} \frac{\Phi_{II}(x,mx)}{||\Phi_{II}(x,mx)||}=\frac{(\rho,\ 1-\rho)}{||(\rho,\ 1-\rho)||}}=:\vec{\vartheta}$.
\end{prop}

In other words, as $ (x,y)\in(\R_+)^2 $ approaches the origin along the straight line $y=mx$ (with positive $m$),
both the vector fields $ \Phi_I $ and $ \Phi_{II} $ tend to have the orientation of the vector $(\rho,\ 1-\rho)$.
The behavior analysis of the vector fields on such a straight line will be used in theorem \ref{noCycles}
to find sufficient conditions under which the vector fields do not admit cycles and spirals.

\section{Local and global behavior}\label{sec.behavior}

Now, we will examine the behavior of the solutions near the singular points. To do that,
we will compute the Jacobian matrices of the vector fields and determine their eigenvalues. Note that the Jacobian matrices of the vector fields $\Phi_I$ and
$\Phi_II$ are, respectively,
\begin{equation*}
\begin{array}{ccc}
D\Phi_I(x,y)=\left[
               \begin{array}{cc}
                 \rho - 2\X_x x    & -2\X_y y \\\\
                 (1-\rho)- 2\Y_x x & -2\Y_y y \\
               \end{array}
             \right]

& \text{ and } &

D\Phi_{II}(x,y)=\left[
               \begin{array}{cc}
                 -2\X_x x & \rho r-2\X_y y \\\\
                 -2\Y_x x & (1-\rho)r-2\Y_y y\\
               \end{array}
             \right].
\end{array}
\end{equation*}

The table summarizes the signs of the trace, determinant and discriminant of the Jacobians $D\Phi_{I}$ and $D\Phi_{II}$ for each singularity of the respective vector field.

\begin{table}[h]
\begin{center}
\tiny
\noindent\begin{tabular}{|p{.65cm}|p{1.35cm}|p{4.25cm}|p{1.5cm}|p{6cm}|}
  \hline
       & \begin{center}Singularity\end{center}  & \begin{center}Trace\end{center} & \begin{center}Determinant\end{center} & \begin{center}Discriminant\end{center}\\
  \hline
  &&&&\\
& $(x_I,y_I)$   & $ sign\Bigl(\rho-2(\X_xx_I+\Y_yy_I)\Bigr)$ & $sign\Bigl(\Delta_y\Bigr)$ & $sign\Bigl([\rho-2(\X_xx_I+\Y_yy_I)]^2-8\Delta x_Iy_I\Bigr)$\\
 &&&&\\
\normalsize$\Phi_I$\tiny & $(x_I,-y_I)$   & $sign\Bigl(\rho-2(\X_xx_I-\Y_yy_I)\Bigr)$ & -$sign\Bigl(\Delta_y\Bigr)$ & $sign\Bigl([(\rho-2(\X_xx_I-\Y_yy_I)]^2+8\Delta x_Iy_I\Bigr)$\\
 &&&&\\
& $(0,0)$   & $ sign(1) $ & 0 & $ sign(1) $ \\
 &&&&\\
  \hline
&&&&\\
& $(x_{II},y_{II})$   & $sign\Bigl((1-\rho)r-2(\X_xx_{II}+\Y_yy_{II})\Bigr)$ & $sign\Bigl(\Delta_y\Bigr)$ & $sign\Bigl([(1-\rho)r-2(\X_xx_{II}+\Y_yy_{II})]^2-8\Delta x_{II}y_{II}\Bigr)$\\
 &&&&\\
\normalsize$\Phi_{II}$\tiny & $(-x_{II},y_{II})$   & $sign\Bigl((1-\rho)r+2(\X_xx_{II}-\Y_yy_{II})\Bigr)$ & -$sign\Bigl(\Delta_y\Bigr)$ & $sign\Bigl([(1-\rho)r+2(\X_xx_{II}-\Y_yy_{II})]^2+8\Delta x_{II}y_{II}\Bigr)$\\
 &&&&\\
& $(0,0)$   & $ sign(1) $ & 0 & $ sign(1) $\\
 &&&&\\
 \hline
\end{tabular}
\caption{\small{Elements for the classification of the singularities of both vector fields $\Phi_I$ and $\Phi_{II}$.}}\label{table1}
\end{center}
\end{table}
\normalsize

We highlight that the signs in the above table do not depend on $r$. In fact, from $x_{II}=\frac{\sqrt{|\Delta_x\Delta_y|}}{|\Delta|}r$ and $y_{II}=-\frac{\Delta_x}{\Delta}r$ we can see that $r$ does not affect
the signs of the trace, determinant and discriminant of $D\Phi_{II}(x_{II},y_{II})$ and $D\Phi_{II}(-x_{II},y_{II})$. The next result follows
directly from such a fact and Table~\ref{table1}.

\begin{theo}\label{LocalCharact_1}
For $\Phi_I$ as well as for $\Phi_{II}$, the singularity types do not depend on $r$. Moreover,
\begin{enumerate}
\item for both vector fields, (0,0) is a nonhyperbolic singularity for which the vector $ (\rho, 1 - \rho) $
defines a repulsive direction in the phase space;
\item $(x_I,-y_I)$ and $(-x_{II},y_{II})$ are saddle points if, and only if, \eqref{OUOU_1} holds;
\item $(x_I,y_I)$ and $(x_{II},y_{II})$ are saddle points if, and only if, \eqref{OUOU_2} holds.
\end{enumerate}
\end{theo}

Note that the previous theorem gives a partial characterization of the local behavior.
The next theorem characterizes the global behavior of solutions at the first quadrant
when $(x_I,y_I)$ and $(x_{II},y_{II})$ are saddle points.

\begin{theo}\label{saddle}
If \eqref{OUOU_2} holds, except on the stable manifold of the respective saddle point,
all solutions of the nonsmooth vector fields $\Phi_I \mathbf{1}_{\R_+^* \times \R_+^*} $ and $\Phi_{II} \mathbf{1}_{\R_+^* \times \R_+^*}$
vanish as time goes to $\infty$.
\end{theo}

\begin{proof} For the vector field $ \Phi_I $, note that $\dot{x}$ vanishes on the ellipse
 \begin{equation}\label{ellipse1}\frac{\left(x-\frac{\rho}{2\X_x}\right)^2}{\frac{\rho^2}{4\X_x^2}}+ \frac{y^2}{\frac{\rho^2}{4\X_x\X_y}}=1,\end{equation}
 while $\dot{y}$ vanishes on the ellipse
\begin{equation}\label{ellipse2}\frac{\left(x-\frac{(1-\rho)}{2\Y_x}\right)^2}{\frac{(1-\rho)^2}{4\Y_x^2}}+ \frac{y^2}{\frac{(1-\rho)^2}{4\Y_x\Y_y}}=1.\end{equation}
Condition \eqref{OUOU_2} implies that the horizontal axis of ellipse \eqref{ellipse1} is greater than the horizontal axis of ellipse \eqref{ellipse2},
and that both ellipses have only three intersection points: $(0,0)$, $(x_I,y_I)$ and $(x_I,-y_I)$.
If condition \eqref{OUOU_2} holds, on the restriction to the first quadrant of the ellipses \eqref{ellipse1} and \eqref{ellipse2}, the vector field $\Phi_I$ looks like shown in figure~\ref{fig.semiellipse1}
(regardless of the parameter values). Therefore, a solution starting at a point at the first quadrant
which is not on the stable manifold of $(x_I,y_I)$ will eventually cross one of the canonical axes (since $(x_I,y_I)$ is the unique singularity in
that quadrant, the vector field does not allow cycles there; besides, outside both ellipses the derivatives $\dot{x}$ and $\dot{y}$ are negative).

The proof for the vector field $\Phi_{II}$ follows the same outline, but using that for $\Phi_{II}$ the derivative $\dot{x}$ vanishes on the ellipse
 \begin{equation}\label{ellipse3} \frac{x^2}{\frac{\rho^2r^2}{4\X_x\X_y}}+\frac{\left(y-\frac{\rho r}{2\X_y}\right)^2}{\frac{\rho^2r^2}{4\X_y^2}}=1\end{equation}
 while $\dot{y}$ vanishes on the ellipse
\begin{equation}\label{ellipse4}\frac{x^2}{\frac{(1-\rho)^2r^2}{4\Y_x\Y_y}}+\frac{\left(y-\frac{(1-\rho) r}{2\Y_y}\right)^2}{\frac{(1-\rho)^2r^2}{4\Y_y^2}}=1,\end{equation}
and that \eqref{OUOU_2} implies that
the vertical axis of ellipse \eqref{ellipse3} is smaller than the vertical axis of ellipse~\eqref{ellipse4} (see figure~\ref{fig.semiellipse2}).
\begin{figure}[!htb]
\begin{center}
\includegraphics[scale=.2]{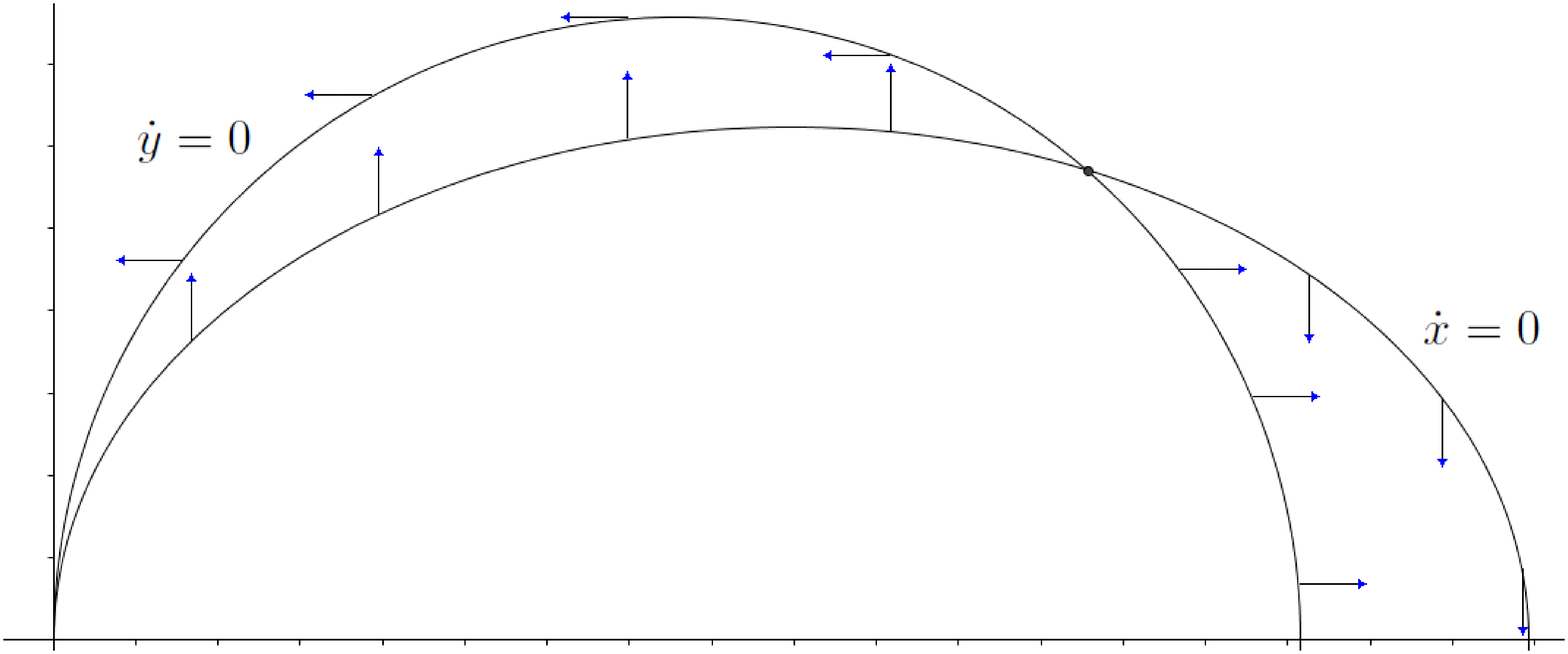}
\caption{\small{Sketch of the behavior of $ \Phi_I $ at the first quadrant.}}
\label{fig.semiellipse1}
\end{center}
\end{figure}
\begin{figure}[!htb]
\begin{center}
\includegraphics[scale=.2]{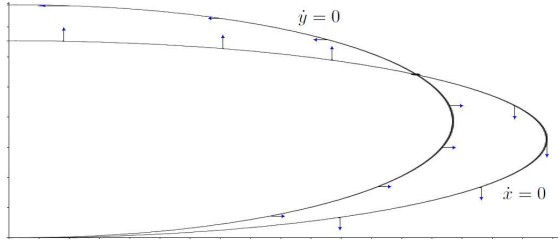}
\caption{\small{Sketch of the behavior of $ \Phi_{II} $ at the first quadrant.}}
\label{fig.semiellipse2}
\end{center}
\end{figure}
\end{proof}

We may notice that the only non-competition parameter on which the singularity type depends is $\rho$. Thus, the secondary sex ratio, $\sigma=(1-\rho)/\rho$, becomes a natural choice of parameter in terms of which we should classify the population behavior. Although the highly nonlinear interdependence of the parameters makes it quite hard to determine the parameter sets that correspond to each possible sign to the entries in Table~\ref{table1}, it is possible to verify that under \eqref{OUOU_1} as $\sigma$ increases the singularity type of $(x_I,y_I)$ changes from stable node to unstable node, passing through stable spiral and unstable spiral. In fact, consider the straight line $\mathcal{L}$ given by $\rho-2(\X_xx+\Y_yy)=0$ and the ellipse $\mathcal{E}$ defined by  $[\rho-2(\X_xx+\Y_yy)]^2-8\Delta xy=0$. Note that $\mathcal{E}$ is contained in the first quadrant of $\R^2$ and touches the axes at the same points that $\mathcal{L}$ crosses the axes, that is, at the points $\left(0, \frac{\rho}{2\Y_y}\right)$ and $\left(\frac{\rho}{2\X_x}, 0\right)$.

Supposing all the parameters are fixed but $\rho$, then $(x_I,y_I)$, $\mathcal{L}$ and $\mathcal{E}$ move on the first quadrant as $\sigma$ changes.
Thus, we can determine the singularity type of $(x_I,y_I)$ by knowing its relative position with respect to $\mathcal{L}$ and $\mathcal{E}$ for each
value of $\sigma$. Indeed,
$$\begin{array}{lcl}
\text{trace}(D\Phi_I(x_I,y_I))<0 &\Longleftrightarrow & \text{$(x_I,y_I)$ is at the right side of $\mathcal{L}$},\\\\
\text{trace}(D\Phi_I(x_I,y_I))=0 &\Longleftrightarrow & \text{$(x_I,y_I)$ is on $\mathcal{L}$},\\\\
\text{trace}(D\Phi_I(x_I,y_I))>0 &\Longleftrightarrow & \text{$(x_I,y_I)$ is at the left side of $\mathcal{L}$},
\end{array}$$
while
$$\begin{array}{lcl}
\text{discriminant}(D\Phi_I(x_I,y_I))<0 &\Longleftrightarrow & \text{$(x_I,y_I)$ is inside the region defined by $\mathcal{E}$},\\\\
\text{discriminant}(D\Phi_I(x_I,y_I))=0 &\Longleftrightarrow & \text{$(x_I,y_I)$ is on $\mathcal{E}$},\\\\
\text{discriminant}(D\Phi_I(x_I,y_I))>0 &\Longleftrightarrow & \text{$(x_I,y_I)$ is outside the region defined by $\mathcal{E}$}.
\end{array}$$

Now, observe that if $\sigma\to \frac{\Y_x}{\X_x}^+$, then $x_I \to \frac{\rho}{\X_x}$ and $y_I \to 0^+$.
According to Table~\ref{table1}, it follows that $\text{trace}(D\Phi_I(x_I,y_I))<0$ and $\text{discriminant}(D\Phi_I(x_I,y_I))>0$, and therefore
$(x_I,y_I)$ is placed at the right side of $\mathcal{L}$ and outside of the region defined by $\mathcal{E}$. On the other hand,
if $\sigma\to \frac{\Y_y}{\X_y}^-$, then $x_I \to 0^+$ and $y_I \to 0^+$, which guarantees that $\text{trace}(D\Phi_I(x_I,y_I))>0$ and
$\text{discriminant}(D\Phi_I(x_I,y_I))>0$, and hence $(x_I,y_I)$ is placed at the left side of $\mathcal{L}$ and outside of the region defined by $\mathcal{E}$.

Since $(x_I,y_I)$, $\mathcal{L}$ and $\mathcal{E}$ move continuously on the first quadrant of $\R^2$ as $\sigma$ changes, then there must exist values of $\sigma$ for which $(x_I,y_I)$ is placed in any relative position with respect to $\mathcal{L}$ and $\mathcal{E}$, that is, generically
$(x_I,y_I)$ can be a singularity of any of the four announced types. Furthermore, it is possible to verify that as $\sigma$ increases the point $(x_I,y_I)$ passes from the right of $\mathcal{L}$ to the left of $\mathcal{L}$, and from the outside of $\mathcal{E}$ to the inside of $\mathcal{E}$ and once again to the outside of $\mathcal{E}$ (see figure~\ref{GlobalAnal}).

By an analogous analysis, comparing the relative position of $(x_{II},y_{II})$ with respect to the straight line $(1-\rho)r-2(\X_xx+\Y_yy)=0$ and the ellipse $[(1-\rho)r-2(\X_xx+\Y_yy)]^2-8\Delta xy=0$, we can verify that,
as $\sigma$ increases, the singularity type of $(x_{II},y_{II})$ changes from unstable node to stable node, passing through unstable spiral and stable spiral.

\begin{center}
\begin{figure}[ht]
\begin{center}
\begin{minipage}[b]{0.4\linewidth}
\includegraphics[width=\textwidth]{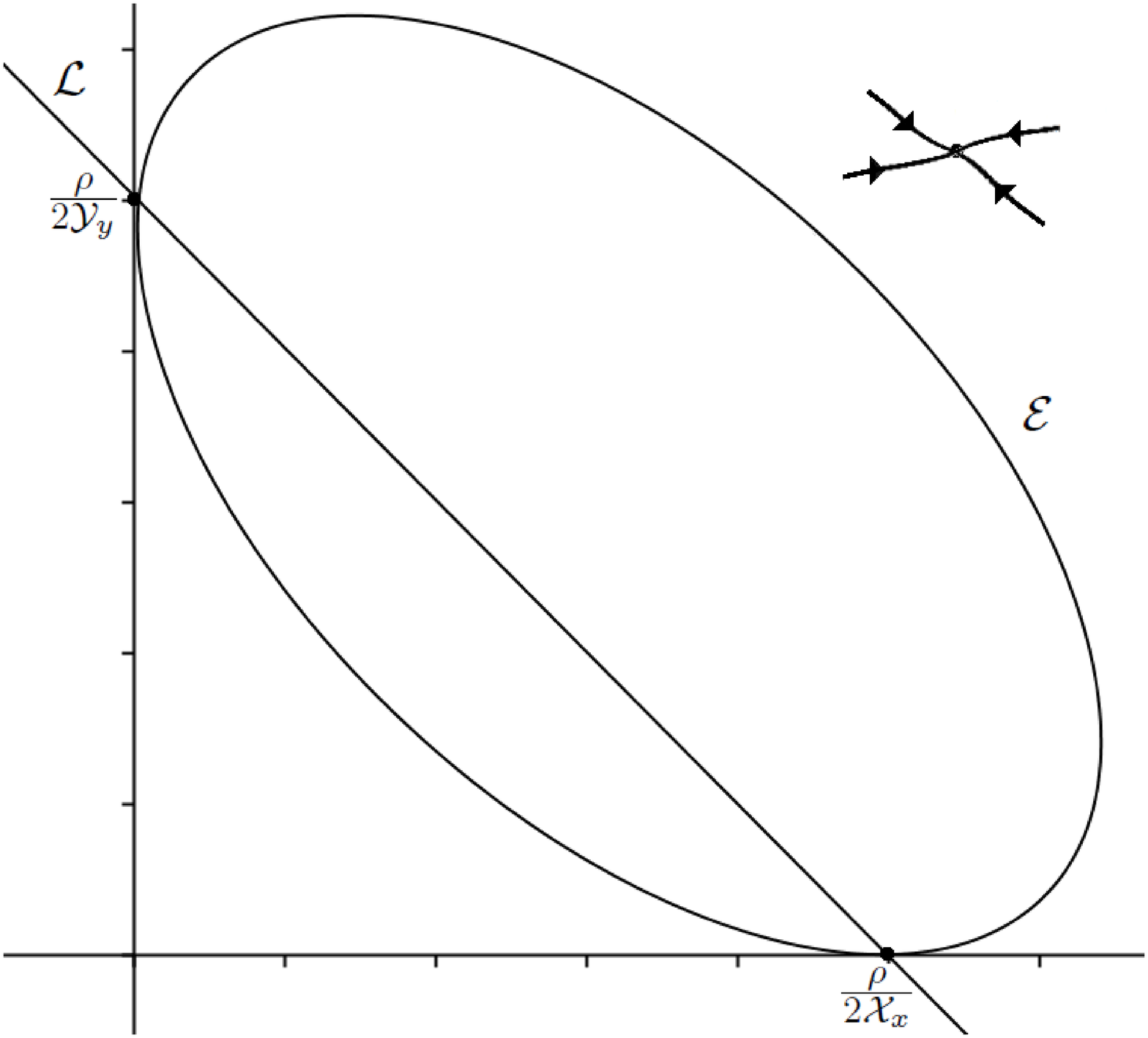}
\end{minipage}
\begin{minipage}[b]{0.4\linewidth}
\includegraphics[width=\textwidth]{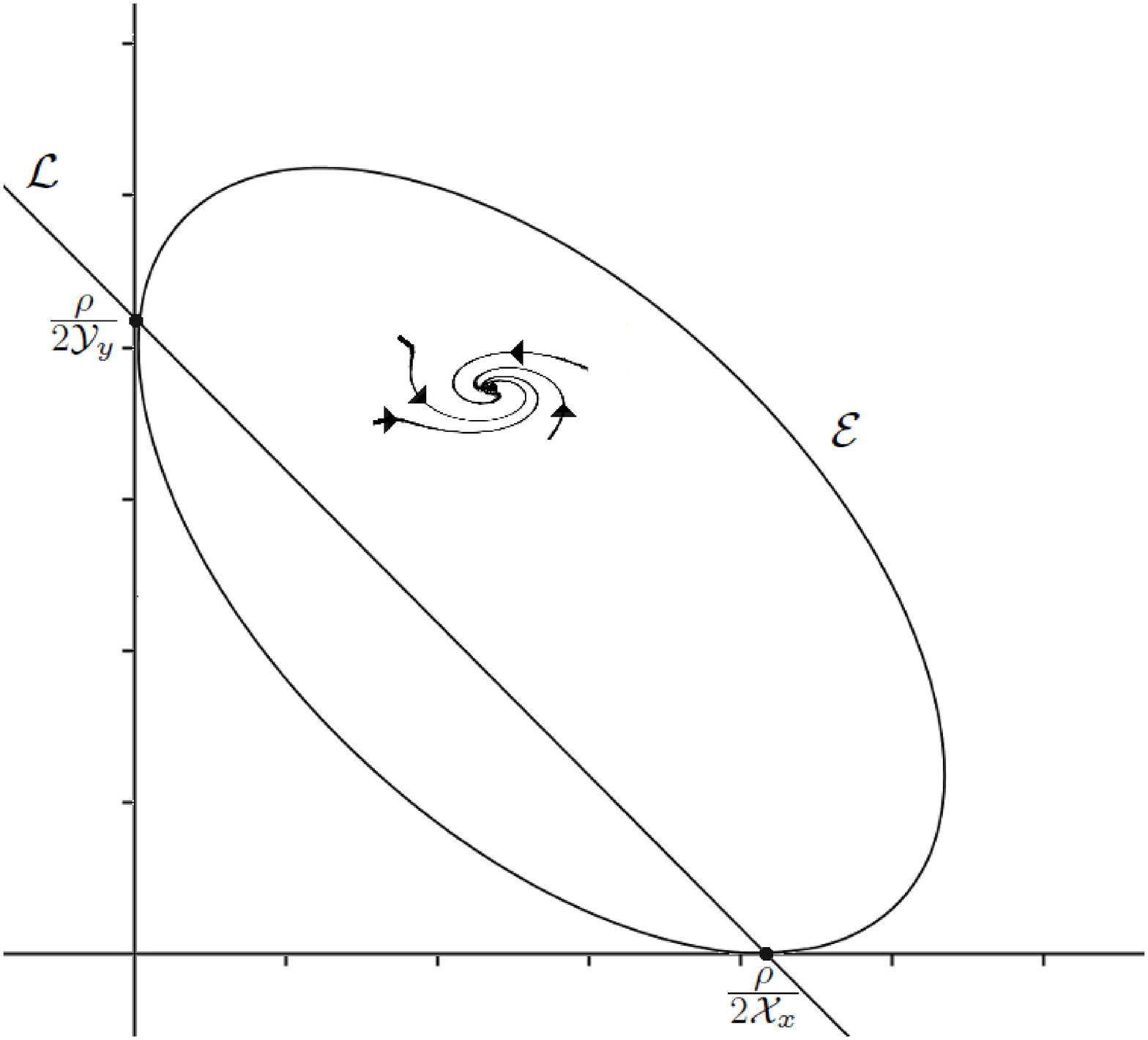}
\end{minipage}\\
\begin{minipage}[b]{0.4\linewidth}
\includegraphics[width=\textwidth]{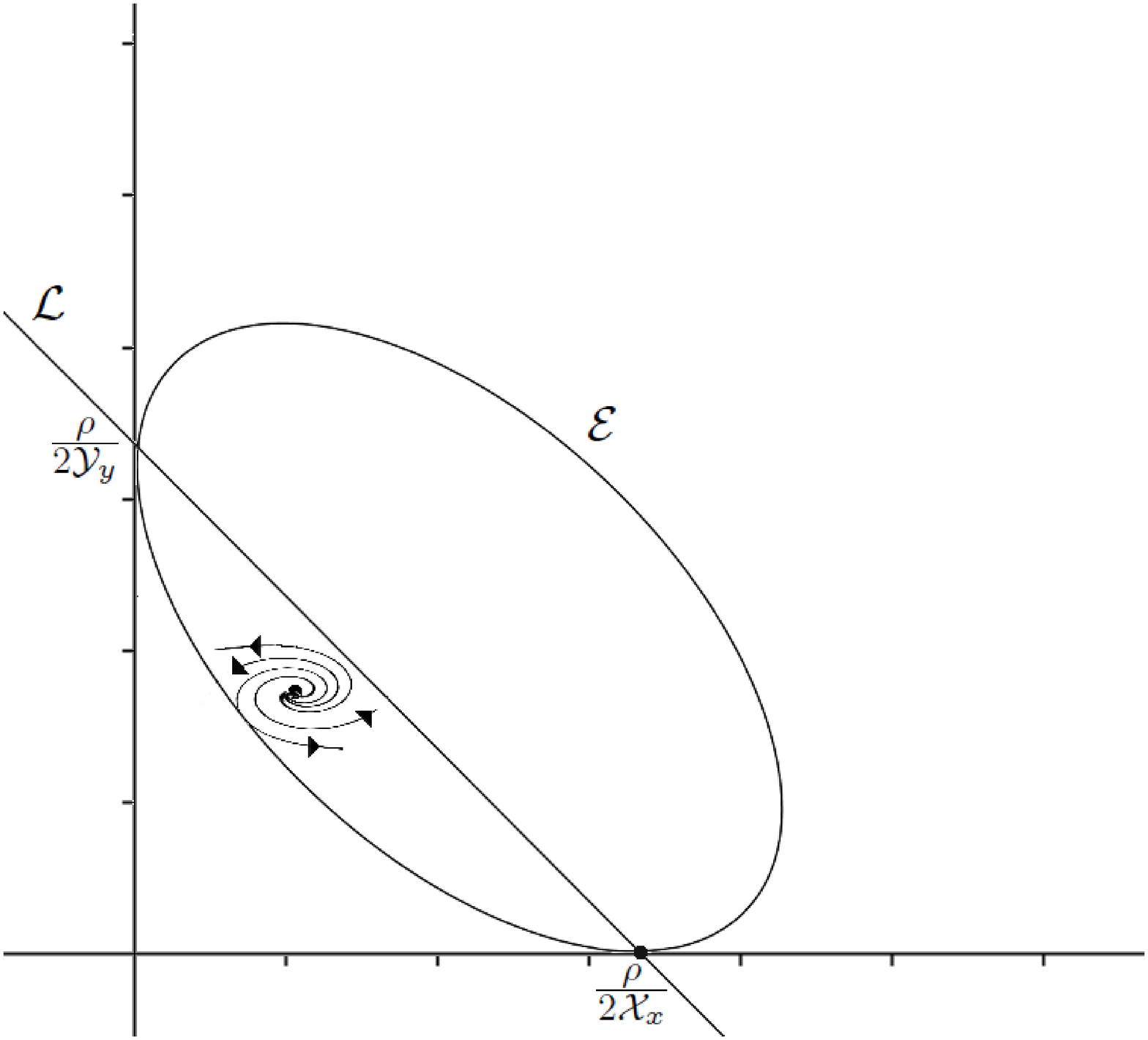}
\end{minipage}
\begin{minipage}[b]{0.4\linewidth}
\includegraphics[width=\textwidth]{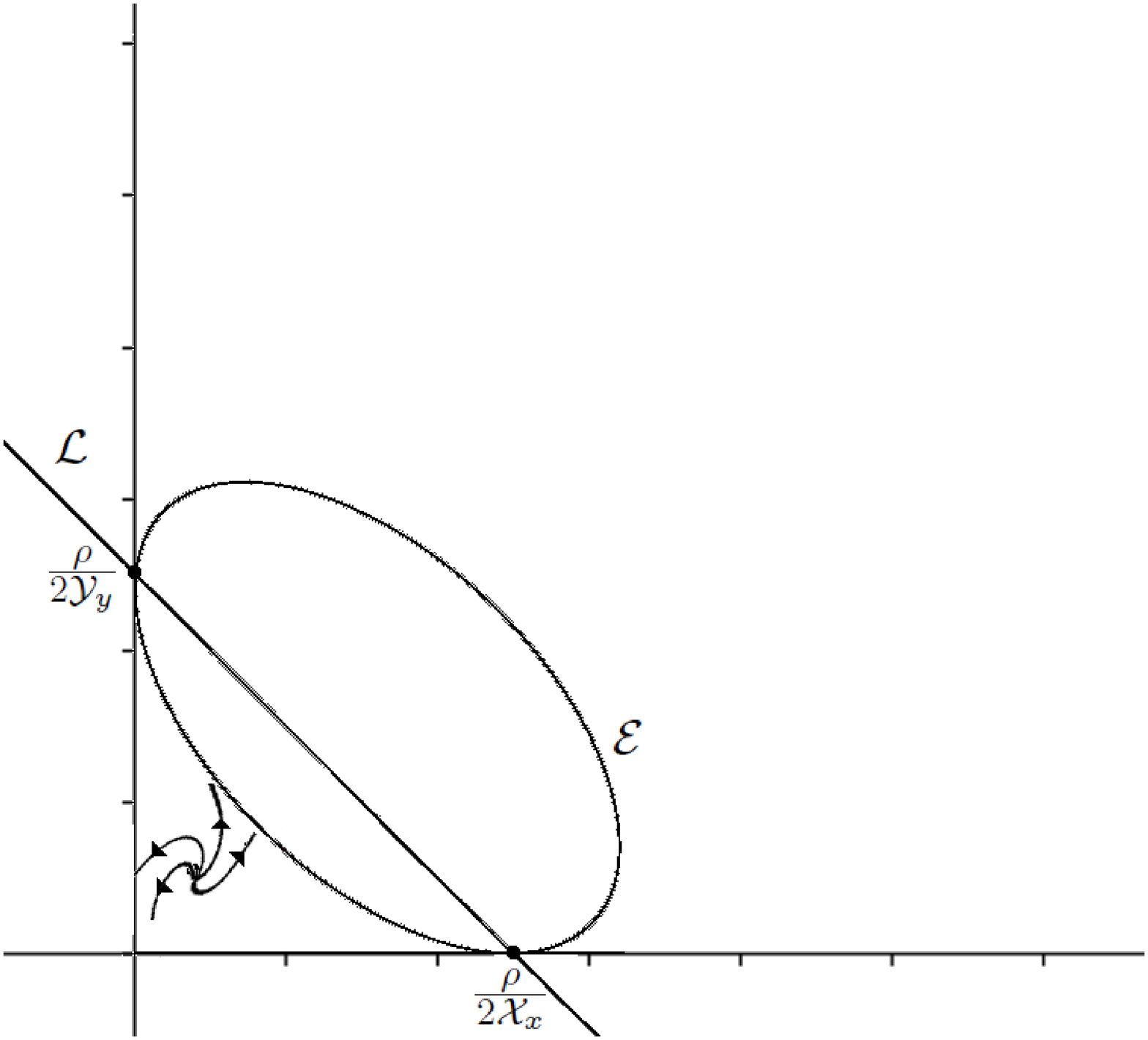}
\end{minipage}
\end{center}
\caption{\small{From left above to right below, as $\sigma$ increases,
the singularity type of $ (x_I, y_I) $ changes from stable node to un\-sta\-ble node,
passing through stable spiral and unstable spiral.}}\label{GlobalAnal}
\end{figure}
\end{center}

The next results present some characterizations of local and global behaviors of positive solutions, based on the relationship between the secondary sex ratio and the competition parameters.

\begin{theo}\label{SingI_II} Under condition \eqref{OUOU_1}, we have that

\smallskip

\noindent I.i.  {\color{white}  } $ {\displaystyle\frac{\Y_x}{\X_x}
< \sigma < \frac{\left(\sqrt{\X_x\X_y+\Y_y^2} + \Y_y\right)^2\Y_y+\Y_x\X_y^2}{\left(\sqrt{\X_x\X_y+\Y_y^2} + \Y_y\right)^2\X_y+\X_x\X_y^2}
\Longrightarrow (x_I,y_I) \text{ is stable (spiral or node)}} $;

\noindent I.ii. {\color{white} } $ {\displaystyle \frac{\left(\sqrt{\X_x\X_y+\Y_y^2} + \Y_y\right)^2\Y_y+\Y_x\X_y^2}{\left(\sqrt{\X_x\X_y+\Y_y^2} + \Y_y\right)^2\X_y+\X_x\X_y^2}
< \sigma < \frac{\Y_y}{\X_y}
\Longrightarrow (x_I,y_I) \text{ is unstable (spiral or node)}} $;

\noindent II.i. {\color{white} } $ {\displaystyle \frac{\Y_x}{\X_x}
< \sigma < \frac{\left(\sqrt{\Y_x\Y_y+\X_x^2} - \X_x\right)^2\Y_y+\Y_x\Y_y^2}{\left(\sqrt{\Y_x\Y_y+\X_x^2} - \X_x\right)^2\X_y+\X_x\Y_y^2}
\Longrightarrow (x_{II},y_{II})\text{ is unstable (spiral or node)}} $;

\noindent II.ii. $ {\displaystyle \frac{\left(\sqrt{\Y_x\Y_y+\X_x^2} - \X_x\right)^2\Y_y+\Y_x\Y_y^2}{\left(\sqrt{\Y_x\Y_y+\X_x^2} - \X_x\right)^2\X_y+\X_x\Y_y^2}
< \sigma < \frac{\Y_y}{\X_y}
\Longrightarrow (x_{II},y_{II}) \text{ is stable (spiral or node)}} $.
\end{theo}

\begin{proof}

From the expression of $\tau$ given in \eqref{tau} we get that \begin{equation}\label{sigma(tau)}\sigma=\frac{\tau^2\Y_y+\Y_x}{\tau^2\X_y+\X_x}.\end{equation} In particular, due to \eqref{OUOU_1}, $\sigma$ is a monotonically increasing function of $\tau^2$. In fact, $$\frac{d\sigma}{d(\tau^2)}=\frac{\Delta}{(\tau^2\X_y+\X_x)^2}>0.$$

To prove {\em I.} we use that
\begin{align*}
sign\Bigl( & \text{trace}(D\Phi_I(x_I,y_I))\Bigr) = sign\Bigl(\rho-2(\X_x x_I+\Y_y y_I)\Bigr) = sign\Bigl(\rho-2(\X_x+\Y_y\tau)x_I\Bigr) \\
& = sign\left(\rho-2(\X_x+\Y_y\tau)\frac{\rho\Y_y-(1-\rho)\X_y}{\X_x\Y_y-\X_y\Y_x}\right) = sign\left(1-2(\X_x+\Y_y\tau)\frac{\Y_y-\sigma\X_y}{\X_x\Y_y-\X_y\Y_x}\right) \\
& =_{(1)} sign\left(1-2(\X_x+\Y_y\tau)\frac{\Y_y-\frac{\tau^2\Y_y+\Y_x}{\tau^2\X_y+\X_x}\X_y}{\X_x\Y_y-\X_y\Y_x}\right)
= sign\Bigl(\tau^2 \X_y-2 \tau \Y_y-\X_x\Bigr),
\end{align*}
where $=_{(1)}$ is due \eqref{sigma(tau)}. Hence, $\text{trace}(D\Phi_I(x_I,y_I))<0$ whenever $0<\tau<\frac{\Y_y+\sqrt{\X_x \X_y+\Y_y^{2}}}{\X_y}$ and $\text{trace}(D\Phi_I(x_I,y_I))>0$ whenever $\tau>\frac{\Y_y+\sqrt{\X_x \X_y+\Y_y^{2}}}{\X_y}$, which by~\eqref{sigma(tau)} conclude the proof.

To prove {\em II.} we use a similar analysis, but considering $(x_{II},y_{II})=r\tau( x_I, y_I)=r\tau( x_I,\tau x_I)$ and~\eqref{sigma(tau)} to deduce that
$$sign\Bigl(\text{trace}(D\Phi_{II}(x_{II},y_{II}))\Bigr)=sign\Bigl(-\tau^2 \Y_y-2 \tau \X_x + \Y_x\Bigr).$$
\end{proof}

Note that theorem \ref{SingI_II} only states conditions on $\sigma$ under which the singularities are stable (or unstable), but it does not specify if the singularity is a node or a spiral. Sufficient conditions under which the singularities are stable nodes and the vector fields do not admit cycles on the
first quadrant will be given in theorem \ref{noCycles}.
In order to prove theorem \ref{noCycles}, we will study the behavior of the vector fields along a straight line $y=mx$ with positive $m$.
Recall that $ Q $ denotes the competition polynomial (see definition~\ref{comp.pol}).
\begin{lem}\label{ChangeOfOrientation}
Given $m>0$, except at the origin,

\medskip

\noindent I.{\color{white}I} $\Phi_I(\hat{x},m\hat{x})$ is collinear to $(1,m)$ only at the point(s) where $Q(m)\hat{x}=\rho m-(1-\rho)$;

\noindent II. $\Phi_{II}(\hat{x},m\hat{x})$ is collinear to $(1,m)$ only at the point(s) where $Q(m)\hat{x}=\rho rm^2-(1-\rho)rm$.
\end{lem}

\begin{proof}
Let $\mathbf{n}\in\R^2$ be a non-null vector orthogonal to $(1,m)$. The results {\em I.} and {\em II.} follow by computing the value of $\hat{x}$ for which $<\Phi_I(\hat{x},m\hat{x}),\mathbf{n}>=0$ and  $<\Phi_{II}(\hat{x},m\hat{x}),\mathbf{n}>=0$, respectively.
\end{proof}

\begin{rem}\label{ObsChangeOfOrientation}
Let $\alpha$ be the real root of $Q$. If $ m\neq \alpha$, then $\Phi_I$ is collinear to $(1,m)$ at $(x_{I}^m,\ y_{I}^m):=(x_{I}^m,\ m x_{I}^m)$,
where $x_{I}^m$ solves the linear equation in lemma \ref{ChangeOfOrientation}.I., while $\Phi_{II}$ is collinear to $(1,m)$ at
$(x_{II}^m,\ y_{II}^m):=(x_{II}^m,\ mx_{II}^m)$, where $x_{II}^m$ solves the linear equation in lemma \ref{ChangeOfOrientation}.II.
In particular, if $\alpha\neq m=\tau$, note that $(x_{I}^\tau,\ y_{I}^\tau)=(x_I,\ y_I)$ and $(x_{II}^\tau,\ y_{II}^\tau)=(x_{II},\ y_{II})$.
On the other hand, if $m=\alpha$, there are two cases: either $\sigma=\alpha=\tau$, which implies that all the points of the straight line $ y = \sigma x $
solve the equations in lemma~\ref{ChangeOfOrientation}, and therefore the intersection of this straight line with the first quadrant is a stable
manifold of both singularities $ (x_I,\ y_I) $ and $ (x_{II},\ y_{II})$; or $\sigma\neq\alpha$, which implies that the origin is the unique point
where $\Phi_{I}$ and $\Phi_{II}$ are collinear
to $(1,m)$.
\end{rem}

\begin{theo}\label{noCycles} Let $\alpha$ be the real root of the competition polynomial $Q$. Then, under condition \eqref{OUOU_1},

\medskip

\noindent  I.{\color{white}{I}} $(x_I,y_I)$ is a stable node and $\Phi_I$ does not admit cycles on the first quadrant if $$\sigma\leq\max\left\{\alpha,\frac{2\Y_x\Y_y}{\X_x\Y_y+\X_y\Y_x}\right\}.$$

\noindent  II. $(x_{II},y_{II})$ is a stable node and $\Phi_{II}$ does not admit cycles on the first quadrant if $$\sigma\geq\min\left\{\alpha,\frac{\X_x\Y_y+\X_y\Y_x}{2\X_x\X_y}\right\}.$$
\end{theo}

\begin{proof} We will prove only {\em I.}, since {\em II.} has an analogous proof.
So let $ H_\tau^+ $ denote the half-plane $\{(x,y)\in\R^2:\ y \ge \tau x\} $.
From lemma \ref{ChangeOfOrientation} and remark \ref{ObsChangeOfOrientation}, when $\sigma=\alpha=\tau$, the intersection of the straight line $y=\tau x$ with the first quadrant is a stable manifold of both singularities, and therefore neither $\Phi_I$ nor $\Phi_{II}$ admit cycles on that quadrant.

If $\sigma<\alpha$, from corollary \ref{AlphaSigmaTau}, we obtain $\sigma>\tau$, and, from lemma \ref{ChangeOfOrientation} and remark \ref{ObsChangeOfOrientation},
we have $(x_{I}^\tau,\ y_{I}^\tau)=(x_{I},\ y_{I})$. On the other hand, proposition \ref{VetcorFieldAtOrigin} guarantees that, on the straight line $y=\tau x$ and near the origin, the vector field $\Phi_I$ has a slope near  $\sigma$, that is, it is pointing to the interior of $H_\tau^+$. Thus, $\Phi_I$ is pointing to the interior of $H_\tau^+$ along all the line segment from the origin to the point $(x_I,y_I)$.

Now, observe that, on the segment of the ellipse \eqref{ellipse2} from the origin to the point $(x_I,y_I)$, the vector field $\Phi_I$ is pointing to the straight line $y = \tau x$ (see figure~\ref{noCyclesFIG1}). Thus, in the bounded region of $ H_\tau^+ $ enclosed by the ellipse \eqref{ellipse2},
the flow will converge to $(x_{I},\ y_{I})$. Since $(x_{I},\ y_{I})$ cannot be a saddle point (theorem \ref{LocalCharact_1}), then it is a stable node. Furthermore, this prevents the existence of a cycle on the first quadrant.

\begin{figure}[!htb]
\begin{center}
\includegraphics[scale=.35]{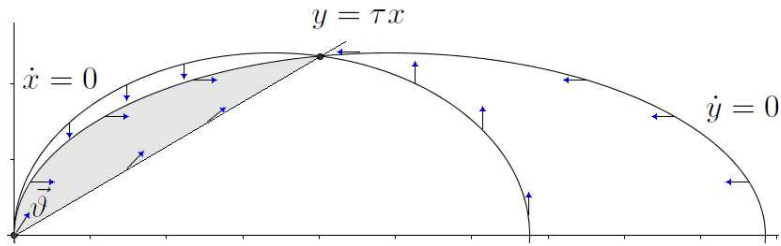}
\caption{\small{Behavior of $\Phi_I$ inside the bounded region of $ H_\tau^+ $ enclosed by the ellipse \eqref{ellipse2}. The arrow labeled as $ \vec{\vartheta}$ in the picture is not a vector of the vector field $\Phi_I$ but represents the direction of $\Phi_I$ near the origin
(see proposition~\ref{VetcorFieldAtOrigin}).}}
\label{noCyclesFIG1}
\end{center}
\end{figure}

For the case when $\sigma\leq\frac{2\Y_x\Y_y}{\X_x\Y_y+\X_y\Y_x}$, we notice that this condition is equivalent to $x_I\geq \frac{1-\rho}{2\Y_x}$, that is, the point $(x_{I},\ y_{I})$ is on or to the right of the vertical axis of the ellipse \eqref{ellipse2}. Thus, analyzing the vector field on the ellipses \eqref{ellipse1} and \eqref{ellipse2}, we see that the shaded region in figure~\ref{noCyclesFIG2} is invariant by the flow. Thus, again we conclude that
$(x_{I},\ y_{I})$ is a stable node and that, on the first quadrant, cycles are not allowed.

\begin{figure}[!htb]
\begin{center}
\includegraphics[scale=.30]{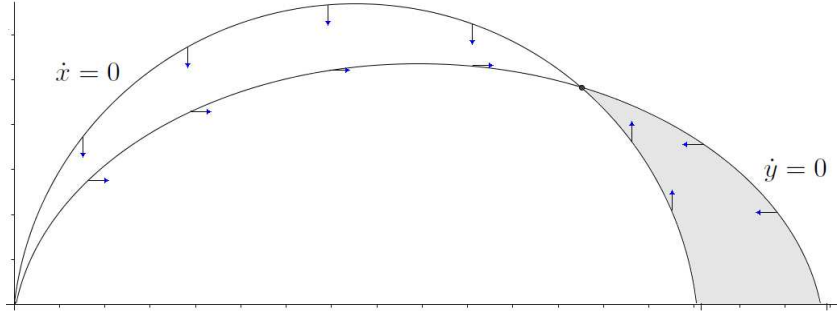}
\caption{\small{Behavior of $\Phi_I$ when $(x_{I},\ y_{I})$ is on or to the right of the vertical axis of the ellipse \eqref{ellipse2}.}}
\label{noCyclesFIG2}
\end{center}
\end{figure}
\end{proof}

\begin{cor}\label{Cor.noCycles}
Under condition \eqref{OUOU_1}, at least one of the singularities $(x_{I},\ y_{I})$ and $(x_{II},\ y_{II})$ is a stable node and its respective vector field does
not admit cycle on the first quadrant.
\end{cor}

\begin{proof}
We only need to show that, if condition $I.$ in theorem \ref{noCycles} does not hold, then condition $II.$ holds.
In order to do that, first note that, under condition \eqref{OUOU_1}, we have $\frac{2\Y_x\Y_y}{\X_x\Y_y+\X_y\Y_x}<\frac{\X_x\Y_y+\X_y\Y_x}{2\X_x\X_y}$.
There are then three possible scenarios:
$$ \alpha\leq\frac{2\Y_x\Y_y}{\X_x\Y_y+\X_y\Y_x}, \qquad
\frac{2\Y_x\Y_y}{\X_x\Y_y+\X_y\Y_x}< \alpha<\frac{\X_x\Y_y+\X_y\Y_x}{2\X_x\X_y}, \qquad
\frac{\X_x\Y_y+\X_y\Y_x}{2\X_x\X_y}\leq \alpha. $$
Thus, if $\sigma>\max\left\{\alpha,\frac{2\Y_x\Y_y}{\X_x\Y_y+\X_y\Y_x}\right\}$, in any case we have that $\sigma>\min\left\{\alpha,\frac{\X_x\Y_y+\X_y\Y_x}{2\X_x\X_y}\right\}$.
\end{proof}

\subsection{Behavior of the flow associated to $\Phi$}\label{sing._of_Phi}

In order to describe the behavior of the flow associated to the original vector field $\Phi$,
we need to understand how the flows associated to the vector fields $\Phi_I$ and $\Phi_{II}$ are ``glued'' along the ray $x=ry$, $y \ge 0$.

First, note that the position of the singular point of $\Phi_I$ does not depend on the parameter $r$. On the other hand, if all parameters remain constant but $r$, as $r$ increases, the coordinates of the singular point of $\Phi_{II}$ also increase. Therefore, from \eqref{Phi_I_II}, this means that, varying $r$, the ray $y=r^{-1}x$, $x \ge 0$, which is the frontier between $R_I$ and $R_{II}$, changes its position, while $(x_{II},y_{II})$ moves along the straight line $y=\tau x$. In particular, note that: if $r^{-1}<\tau$, then the ray $y=\tau x$, $x \ge 0$, is inside of $R_I$ and the non-null singularity of $\Phi$ is $(x_I,y_I)$;
if $r^{-1}=\tau$, then the ray $y=\tau x$, $x \ge 0$, coincides with the frontier between $R_I$ and $R_{II}$, and $(x_I,y_I)=(x_{II},y_{II})$
is a singularity of $\Phi$; and if $r^{-1}>\tau$, then the ray $y=\tau x$, $x \ge 0$, is inside of $R_{II}$ and the non-trivial singularity of $\Phi$ is $(x_{II},y_{II})$.

It is interesting to observe that the average number of the male's reproductive partners $r$ plays a key role in the selection between $(x_I,y_I)$ and $(x_{II},y_{II})$ as singularity of the vector field $\Phi$. Due to corollary~\ref{Cor.noCycles}, under condition \eqref{OUOU_1}, it is always possible to use $r$ to select a stable node as the singularity of $\Phi$.
In fact, if both $(x_I,y_I)$ and $(x_{II},y_{II})$ are stable for their respective vector fields, then the singularity of $\Phi$ will be stable regardless of the value of $r$. However, if $r<\tau^{-1}$, then the non-null singularity of $\Phi$ is $(x_{II},y_{II})$, which satisfies $x_{II}<x_I$ and $y_{II}<y_I$. Thus, the maximum size of an equilibrium population is achieved when $r\geq\tau^{-1}$.

Under condition~\eqref{OUOU_1}, when the secondary sex ratio $\sigma$ is near $\Y_x/\X_x$,
 $(x_I,y_I)$ is a stable node while $(x_{II},y_{II})$ is an unstable singularity. In such a situation, the singularity of $\Phi$ will be stable if, and only if, $r> \tau^{-1}$, that is, when the average number of female sexual partners of each male is larger than the female:male ratio, which means that, on average, all the females are reproducing. On the other hand, if the secondary sex ratio $\sigma$ is near $\Y_y/\X_y$,
then $(x_I,y_I)$ is an unstable singularity, while $(x_{II},y_{II})$ is a stable node. In this case, the singularity of $\Phi$ will be stable if, and only if, $r< \tau^{-1}$, which means that, on average, all the males are reproducing but not all the females are reproducing.

We notice that the above analysis enlighten an interesting feature of the population's equilibrium. Suppose, for an easier comprehension, that $\Y_x/\X_x \ll 1 \ll \Y_y/\X_y$. Therefore, if there are much less males than females being born, then the conservation of the two-sex species depends, in a fundamental way, on the fact that all the females are reproducing successfully. On the other hand, if there are much less females than males being born, then the population will only remain stable and achieve its equilibrium point when a number of females are not reproducing. This apparently contradictory interpretation indicates that the average number of male's reproductive partners $r$ may artificially increase the effect of the competition (with respect to its impact on the population growth) of the female population when this gender has relatively few individuals, allowing the population to reach a stable equilibrium.

The next result presents sufficient conditions for nonexistence of cycles for the flow associated to $\Phi$. In fact, theorem~\ref{noCycles} provides us
conditions for which the vector fields $\Phi_I$ and $\Phi_{II}$ do not admit cycles, but $\Phi$ may have a cycle composed by parts
of orbits which are not cycles for those vector fields (see figure~\ref{cyclesPhi}(b)).

\begin{theo}
The vector field $\Phi$ does not admit cycles if one of the following conditions holds:

\medskip

\noindent I.{\color{white}I}\quad $\Phi_I$ does not admit cycles on the first quadrant and $r^{-1}\leq\min\{\tau,\sigma\}$;

\noindent II.\quad $\Phi_{II}$ does not admit cycles on the first quadrant and $r^{-1}\geq\max\{\tau,\sigma\}$.
\end{theo}

\begin{proof}
Supposing $r^{-1}\leq\min\{\tau,\sigma\}$, since $r^{-1}\leq \tau$, the non-null singularity of $\Phi$ is $(x_I,y_I)$.
If $\Phi_I$ does not admit cycles, then there are no cycles for flow associated to $\Phi$ within the region $R_I$. Due to Poincar\'{e}-Bendixon theorem for non-differentiable vector fields (see, for instance, \cite{Melin}), inside the region enclosed by a periodic orbit there must be at least one singularity. Since $\Phi$ is null outside the first quadrant, the unique possibility is that a cycle for $\Phi$  must pass from $R_I$ to $R_{II}$ and then return to $R_I$ going around $(x_I,y_I)$. If $\sigma=\tau$, from corollary~\ref{AlphaSigmaTau} and remark~\ref{ObsChangeOfOrientation}, the ray $y=\tau x$, $ x \ge 0 $, is a stable manifold of  $(x_I,y_I)$, and hence $\Phi$ does not admit cycles. Thus, suppose that $ \sigma \neq \tau $.
Notice that the existence of a cycle implies that $\Phi$ changes its orientation with respect to the regions $R_I$ and $R_{II}$ on the ray $y=r^{-1}x$, $x \ge 0 $. From lemma~\ref{ChangeOfOrientation}, this can only happen at the origin and at the point $A:=(\hat{x},r^{-1}\hat{x})$ such that
$Q(r^{-1})\hat{x}=\rho r^{-1}-(1-\rho)$. Let then $B$ and $C$ denote, respectively, the points where the straight line $y=r^{-1}x$ intercepts the ellipse~\eqref{ellipse1} and the ellipse~\eqref{ellipse2}. Since $r^{-1}\leq\min\{\tau,\sigma\}$, we prove {\em I.} by analyzing the following cases:

\begin{description}

\item[$r^{-1}=\sigma<\tau$:] This means that $A$ coincides with the origin, and thus $\Phi$ does not change its orientation with respect to the regions $R_I$ and $R_{II}$ along the ray $y=r^{-1}x$, $ x \ge 0 $, which prevents $ \Phi $ to have a cycle.

\item[$r^{-1}<\sigma$:] From proposition \ref{VetcorFieldAtOrigin}, when $ (x, y) $ approaches the origin along the straight line $y=r^{-1}x$, the vector
$\Phi(x,y)$ tends to have the orientation of $(\rho,1-\rho)$, and thus it is pointing to inside of region $R_I$. Therefore, since at the points $B$ and $C$ the vector field $\Phi$ is also pointing to inside of $R_I$, this shows that $\Phi$ points to inside of $R_I$ along all the line segment from the origin to $C$. Otherwise, the vector field would change at least twice its orientation with respect to the regions $R_I$ and $R_{II}$, but it can only change at point $A$ (see figure~\ref{cyclesPhi}(a)). Such a configuration clearly prevents the existence of a cycle for $ \Phi $.
\end{description}
\begin{center}
\begin{figure}[ht]
\begin{center}
\begin{minipage}[b]{0.48\linewidth}
\includegraphics[width=\textwidth]{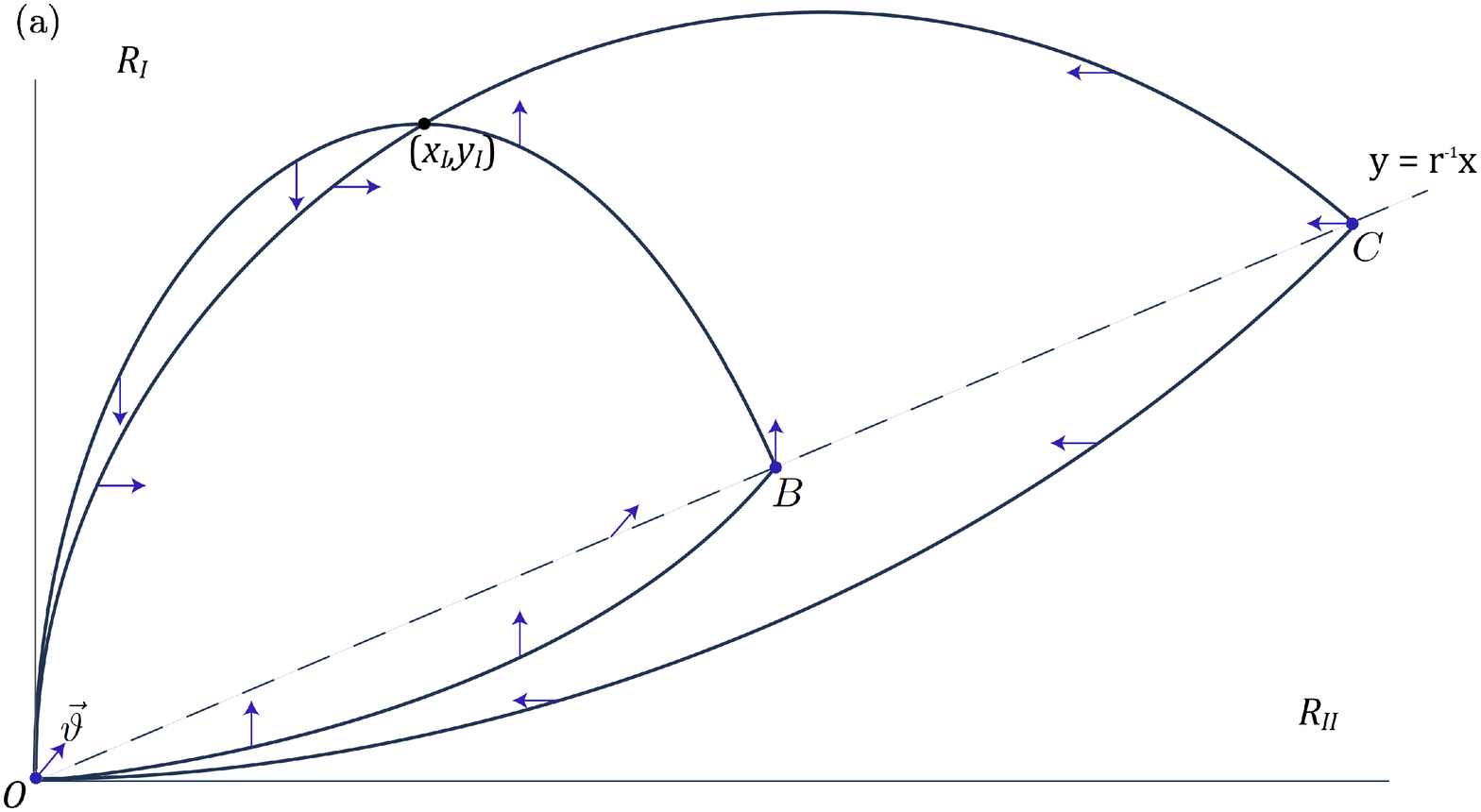}
\end{minipage}
\begin{minipage}[b]{0.48\linewidth}
\includegraphics[width=\textwidth]{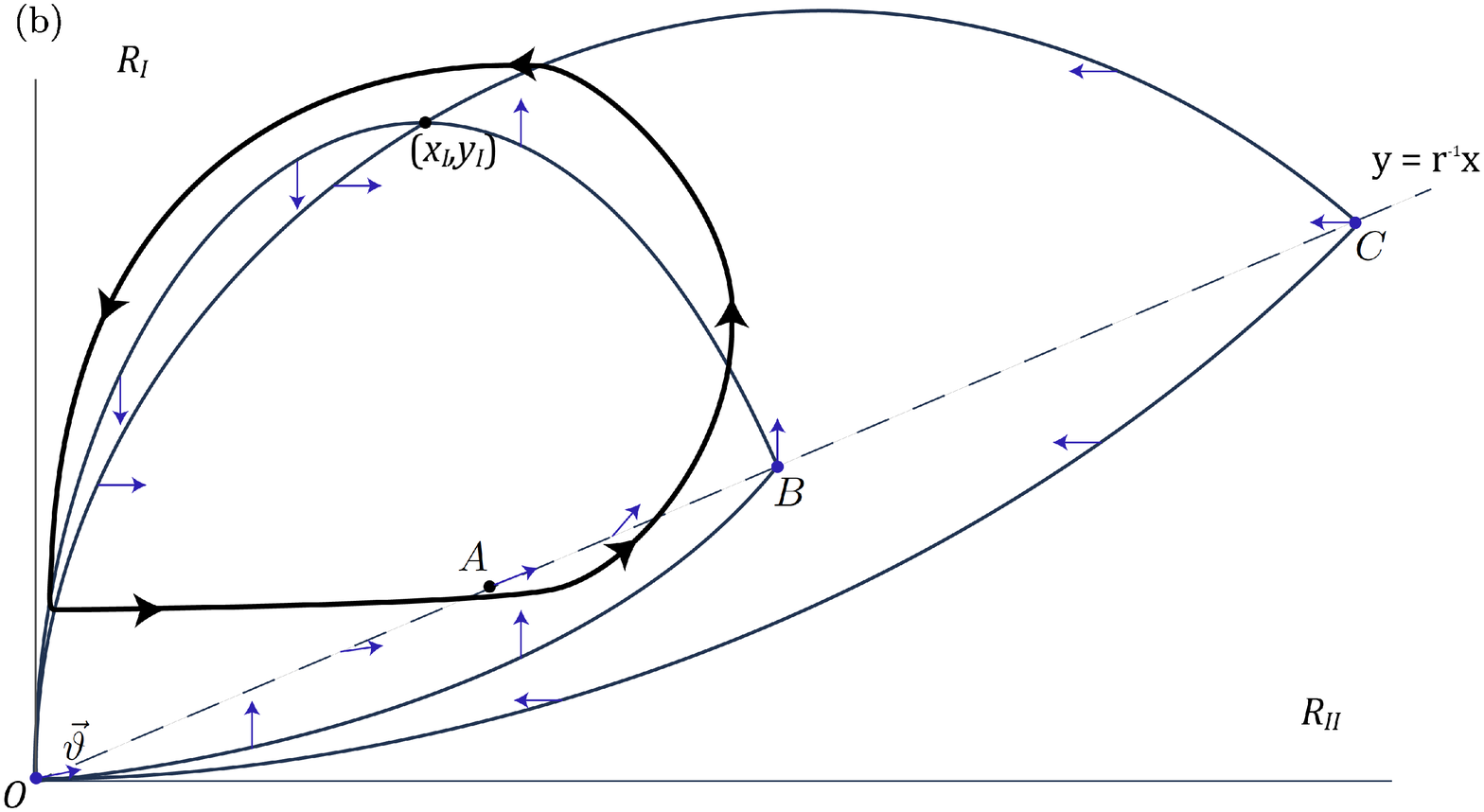}
\end{minipage}\\
\end{center}
\caption{\small{ (a) A sketch of the vector field $\Phi$ when $r^{-1}<\min\{\tau,\sigma\}$: in this case, the orbit cannot pass from $R_I$ to $R_{II}$ along the line segment from the origin to the point $ C $. (b) A sketch of the vector field $\Phi$ and of a possible cycle when $\sigma<r^{-1}<\tau$: under this condition, the point $A$ given by lemma \ref{ChangeOfOrientation} lies on the line segment from the origin to point $B$. In both pictures,
the arrow labeled as $ \vec{\vartheta}$ represents the direction of $\Phi_I$ near of origin (see proposition \ref{VetcorFieldAtOrigin})}.}\label{cyclesPhi}
\end{figure}
\end{center}

The proof of {\em II.} is analogous.
\end{proof}

\section{Final discussion}\label{sec.final}

We have considered a two-sex logistic model given by a vector field that is non-differentiable on a straight line parameterized by the average number of female sexual partners of each male, and in which the growth of each gender is negatively affected by inter-, intra- and outer-gender competitions. Adopting a generic point of view, we have shown that, in the case without inter-gender competition and with mortality rates negligible with respect to the density-dependent mortalities, the population is persistent only if the secondary sex-ratio and competition parameters satisfy specific inequalities (condition \eqref{OUOU_1}), which reflect in particular that the effects of male-male competitions will have relatively greater impact on the male population than on the female population, while the effects of female-female competitions will have relatively greater impact on the female population than on the male population. Furthermore, we have argued that the average number of male's reproductive partners could be seen as an adjustable parameter which may allow a two-sex species to find a stable equilibrium for a large set of secondary sex ratios and competition parameters.

A question that remains open is whether there exist parameters for which the flow of the vector field $\Phi$ has cycles.
Besides, it also remains open to analyze the behavior of the model with all the parameters being non-null, which should reveal a richer dynamics.

\vspace{-.1cm}
\acknowledgement{The authors thank the mathematics departments of both UNICAMP and UFSC for the hospitality during the preparation of this manuscript, and their graduate programs for the financial support. The authors thank Eduardo da Veiga Beltrame for helping us with the preparation of figure 6(b). M. Sobottka was partially supported by CNPq-Brazil grant 304813/2012-5. E. Garibaldi was partially supported by CNPq-Brazil grant 306177/2011-0.}


\vspace{-.1cm}
{\footnotesize
\begin{thebibliography}{99}

\bibitem{Charnov}
E. L. Charnov, \emph{The theory of sex allocation}, Princeton University Press, New Jersey, 1982.

\bibitem{ChavezHuang}
C. Castillo-Chavez and W. Huang, The logistic equation revisited: the two-sex case,
{\em Mathematical Biosciences} \textbf{128} (1995), 299-316.

\bibitem{Fisher}
R. A. Fisher, \emph{The genetical theory of natural selection},
Clarendon Press, Oxford, 1930.

\bibitem{Fredrickson}
A. G. Fredrickson,
A mathematical theory of age structure in sexual populations: random mating and monogamous marriage models,
{\em Mathematical Biosciences} \textbf{10} (1971), 117-143.

\bibitem{Galliard_et_al}
J. F. Le Galliard, P. S. Fitze, R. Ferri\`{e}re and J. Clobert,
Sex ratio bias, male aggression, and population collapse in lizards,
\emph{Proceedings of the National Academy of Sciences of the United States of America} \textbf{102} (2005), 18231-18236.

\bibitem{Garibaldi-Sobottka}
E. Garibaldi and M. Sobottka, Average sex ratio and population maintenance cost,
\emph{SIAM Journal on Applied Mathematics} \textbf{71} (2011), 1009-1025.

\bibitem{Goodman} L. A. Goodman,
Population growth of the sexes,
{\em Biometrics} \textbf{9} (1953), 212-225.

\bibitem{Grayson_et_al}
K. L. Grayson, S. P. De Lisle, J. E. Jackson, S. J. Black and E. J. Crespi,
Behavioral and physiological female responses to male sex ratio bias in a pond-breeding amphibian,
\emph{Frontiers in Zoology} \textbf{9} (2012), 1-10.

\bibitem{HWW}
K. P. Hadeler, R. Waldst\"atter and A. W\"orz-Busekros,
Models for pair formation in bisexual populations,
\emph{Journal of Mathematical Biology} \textbf{26} (1988), 635-649.

\bibitem{Hamilton}
W. D. Hamilton, Extraordinary sex ratios, {\em Science} \textbf{156} (1967), 477-488.

\bibitem{HeskethXing}
T. Hesketh and Z. W. Xing, Abnormal sex ratios in human populations: causes and consequences,
{\em Proceedings of the National Academy of Sciences of the United States of America} \textbf{103} (2006), 13271-13275.

\bibitem{Hoppensteadt}
F. Hoppensteadt, \emph{Mathematical theory of populations: demographics, genetics and epidemics},
SIAM, Philadelphia, 1975.

\bibitem{Kendall49}
D. G. Kendall, Stochastic processes and population growth, \emph{Journal of the Royal Statistical Society. Series B (Methodological)} \textbf{11} (1949), 230-264.

\bibitem{Melin}
J. Melin, Does distribution theory contain means for extending Poincar\'e-Bendixon theory?,
\emph{Journal of Mathematical Analysis and Applications} \textbf{303} (2004), 81-89.

\bibitem{Michler_et_al11}
S. P. M. Michler, M. Nicolaus, R. Ubels, M. van der Velde, C. Both, J. M. Tinbergen and J. Komdeur,
Do sex-specific densities affect local survival of free-ranging great tits?, \emph{Behavioral Ecology} \textbf{22} (2011), 869-879.

\bibitem{Ranta_et_Al}
E. Ranta, V. Kaitala and J. Lindstr\"{o}m,
Sex in space: population dynamic consequences,
{\em Proceedings of Royal Society of London, Series B: Biological Sciences} \textbf{266} (1999), 1155-1160.

\bibitem{Rosen83}
K. H. Rosen,
Mathematical models for polygamous mating systems,
{\em Mathematical Modelling} \textbf{4} (1983), 27-39.

\bibitem{TriversWillard}
R. L. Trivers and D. E. Willard, Natural selection of parental ability to vary the sex ratio of offspring, \emph{Science} \textbf{179} (1973), 90-92.

\bibitem{YangMilner} K. Yang and F. Milner, The logistic, two-sex, age-structured population model,
{\em Journal of Biological Dynamics} \textbf{3} (2009), 252-270.

\bibitem{YellinSamuelson} J. Yellin and P.~A. Samuelson,
A dynamical model for human population,
{\em Proceedings of the National Academy of Sciences of the United States of America} \textbf{71} (1974), 2813-2817.
\end{thebibliography}
}
\end{document}